\documentclass[conference]{IEEEtran}

\usepackage{bm}
\usepackage{version}
\usepackage{amsmath, amssymb}
\usepackage[pdftex]{color}
\usepackage[pdftex]{graphicx}
\usepackage{tabularx}
\usepackage{booktabs}
\usepackage{url}
\usepackage{cases}
\usepackage{amsthm}
\usepackage{textcomp}

\allowdisplaybreaks[4]

\newcommand{\tbm}[1]{\tilde{\bm{#1}}}
\newcommand{\tsf}[1]{\tilde{\mathsf{#1}}}

\newcommand{\transpose}[1]{#1^{\mathsf{T}}}
\newcommand{\ctranspose}[1]{#1^{\mathsf{H}}}

\newcommand{\argmin}{\mathop{\mathrm{arg~min}}\limits}

\newtheorem{Lem}{Lemma}

\makeatletter
\def\ps@IEEEtitlepagestyle{%
    \def\@oddfoot{\mycopyrightnotice}%
    \def\@evenfoot{}%
}
\def\mycopyrightnotice{%
    {\footnotesize Preprint Manuscript of 2017 IEEE International Workshop on Computational Advances in Multi-Sensor Adaptive Processing (CAMSAP 2017) \textcopyright2017 IEEE\hfill}
    \gdef\mycopyrightnotice{}
}
\makeatother

\ifCLASSINFOpdf
\else
\fi
\ifCLASSOPTIONcompsoc
 \usepackage[caption=false,font=normalsize,labelfont=sf,textfont=sf]{subfig}
\else
 \usepackage[caption=false,font=footnotesize]{subfig}
\fi
\begin{document}
\bstctlcite{IEEEexample:BSTcontrol}

%
\title{Independent Low-Rank Matrix Analysis Based on Parametric Majorization-Equalization Algorithm}



%
\author{\IEEEauthorblockN{Yoshiki Mitsui\IEEEauthorrefmark{1},
Daichi Kitamura\IEEEauthorrefmark{1},
Norihiro Takamune\IEEEauthorrefmark{1}, \\
Hiroshi Saruwatari\IEEEauthorrefmark{1},
Yu Takahashi\IEEEauthorrefmark{2},
and Kazunobu Kondo\IEEEauthorrefmark{2}}
\IEEEauthorblockA{\IEEEauthorrefmark{1}The University of Tokyo, Tokyo, Japan}
\IEEEauthorblockA{\IEEEauthorrefmark{2}Yamaha Corporation, Shizuoka, Japan}}

\maketitle

%
\IEEEpeerreviewmaketitle

\begin{abstract}
    In this paper, we propose a new optimization method for independent low-rank matrix analysis (ILRMA)
    based on a parametric majorization-equalization algorithm. ILRMA is an efficient blind source
    separation technique that simultaneously estimates a spatial demixing matrix (spatial model)
    and the power spectrograms of each estimated source (source model). In ILRMA, since both models are
    alternately optimized by iterative update rules, the difference in the convergence speeds between
    these models often results in a poor local solution. To solve this problem, we introduce a new parameter
    that controls the convergence speed of the source model and find the best balance between the optimizations
    in the spatial and source models for ILRMA.
\end{abstract}


\section{Introduction}
Audio source separation is a technique for estimating individual audio sources from a mixture signal.
In the last decade, nonnegative matrix factorization (NMF)~\cite{Lee1999,Lee2001,Kitamura2016IWAENC}
has been utilized to solve the audio source separation problem for single-channel signals.
In particular,
Itakura--Saito NMF (ISNMF) has a special
interpretation of a generative model that justifies the decomposition of power spectrograms~\cite{Fevotte2009},
and it has been extended
to a multichannel format,
which is known as multichannel NMF (MNMF)~\cite{Ozerov2010,Kameoka2010,Sawada2013,Mitsufuji2016}.
In a fully blind situation, blind source separation (BSS) has been well investigated, where the separated
signals are estimated without knowing any prior infomation about the mixing conditions or the characteristics of
the sources.
In the BSS problem, methods based on statistical independence, such as
frequency-domain independent component analysis (ICA)~\cite{Smaragdis1998,Sawada2004,Saruwatari2006} and
independent vector analysis (IVA)~\cite{Kim2007}, are widely used.

In these works, parameters are often estimated by an auxiliary function technique,
which is a computer-intensive method for iteratively solving various optimization problems
with a full guarantee of the monotonic convergence.
This includes the \emph{majorization-minimization (MM) algorithm}~\cite{Hunter2000}
that is a generalized scheme of a well-known EM algorithm,
and its improved \emph{majorization-equalization (ME) algorithm}~\cite{Fevotte2011}.
For example, MM- and ME-algorithm-based multiplicative update rules have been derived for ISNMF,
and it has been reported that the optimization based on the ME algorithm results in faster convergence
than that based on the MM algorithm~\cite{Fevotte2011,Nakano2010,Cao1999}.
For the BSS methods, the MM algorithm has been utilized to derive fast and stable update rules of
ICA~\cite{Ono2010} and IVA (AuxIVA)~\cite{Ono2011}.

Recently, a unified algorithm of IVA and ISNMF for BSS,
\emph{independent low-rank matrix analysis (ILRMA)}~\cite{Kitamura2016,Mitsui2017}, has been proposed.
In this method, parameters of a spatial demixing matrix (spatial model) and NMF variables that represent
the power spectrograms of each source (source model) are alternately optimized in the separate MM algorithms.
From a mathematical viewpoint, this scheme is classified into the coordinate descent method
and needs special care for control of each separated optimization because the ILRMA's problem is non-convex.
It is empirically known that the convergence speed of the spatial model is much higher than that of the
source model, and this difference may result in unstable source separation
owing to the existence of local minima in the cost function.
Therefore, the convergence speeds of these models should be controlled by introducing a new parameter
to find the best balance between the optimizations in the spatial and source models.

In this paper, we propose a new optimization method for ILRMA that is based on the parametric ME algorithm
with the additional parameter that controls the convergence speed of the source model.
Our new findings on the parametric surrogate function for
an arbitrary power term enable us to construct various types of
auxiliary functions that can be solved with a closed form in the ME algorithm step.
It would be of great interest in BSS that the proposed method is a generalization
of the conventional optimization algorithms used in ILRMA and NMFs;
indeed our algorithm includes conventional ILRMA as a special case.
Also, the optimal balance between the convergence speeds is experimentally confirmed by BSS experiments using
music signals.

\section{Conventional Methods}
\label{sec:conventional}

\subsection{Blind Source Separation in Frequency Domain}
Let $N$ and $M$ be the numbers of sources and microphones, respectively.
The complex-valued short-time Fourier transform (STFT) coefficients
of source signals, observed signals, and separated signals are defined as
$\bm{s}_{ij} = \transpose{(s_{ij, 1},\ldots, s_{ij, n}, \ldots, s_{ij, N})}$,
$\bm{x}_{ij} = \transpose{(x_{ij, 1},\ldots, x_{ij, m}, \ldots, x_{ij, M})}$,
and $\bm{y}_{ij} = \transpose{(y_{ij, 1},\ldots, y_{ij, n}, \ldots, y_{ij, N})}$,
where $i=1,\ldots,I$; $j=1,\ldots,J$; $n=1,\ldots,N$; and $m=1,\ldots,M$ are
the integral indexes of the frequency bins, time frames, sources, and channels, respectively,
and $\transpose{}$ denotes a transpose.
When the window length in the STFT is sufficiently longer than
the impulse responses between sources and microphones,
the following instantaneous mixture model in a frequency domain holds:
\begin{align}
  \bm{x}_{ij} = \bm{A}_{i}\bm{s}_{ij}, \label{eq:rank1}
\end{align}
where $\bm{A}_{i}$ is the mixing matrix.
This assumption of the mixing system is often called a rank-1 spatial model~\cite{Duong2010}.
If the number of sources equals the number of channels ($M=N$),
the demixing matrix $\bm{W}_{i} = \ctranspose{(\bm{w}_{i, 1}, \ldots, \bm{w}_{i, N})} = \bm{A}_{i}^{-1}$
can be defined, where $\ctranspose{}$ denotes a Hermitian transpose, and the separated signals are represented as
\begin{align}
    \bm{y}_{ij} = \bm{W}_{i} \bm{x}_{ij}.
\end{align}
The BSS problem is to estimate $\bm{W}_i$ using only $\bm{x}_{ij}$.

\subsection{ILRMA}
ILRMA is a method unifying IVA and ISNMF, which allows us to simultaneously model
the statistical independence between sources and the sourcewise time-frequency structure.
Whereas IVA employs sourcewise frequency vectors~\cite{Kim2007},
ILRMA estimates sourcewise power spectrograms that are approximately modeled by the NMF variables.
The demixing matrix $\bm{W}_i$ and the separated signal $\bm{y}_{ij}$ are optimized so that
the spectrogram of each source tends to be a low-rank matrix.
This separation mechanism of ILRMA is shown in Fig.~\ref{fig:ILRMA},
where $\bm{T}_{n} \in \mathbb{R}_{\geq 0}^{I{\times}L}$ and
$\bm{V}_{n} \in \mathbb{R}_{\geq 0}^{L{\times}J}$ are the basis and activation matrices for the $n$th estimated source,
respectively, and $l = 1,\ldots,L$ is the integral index of the bases.
$\bm{W}_i$, $\bm{T}_n$, and $\bm{V}_n$ can consistently be estimated in a fully blind manner.
Note that ILRMA is theoretically equivalent to conventional MNMF only when the rank-1 spatial model \eqref{eq:rank1}
is assumed, which yields a more stable and more efficient algorithm.

The cost function in ILRMA is defined as follows:
\begin{align}
  \mathcal{J}(\mathsf{W},\mathsf{T},\mathsf{V}) &= \sum_{i, j, n}  \left[
  \frac{| \ctranspose{\bm{w}_{i,n}}\bm{x}_{ij} |^2}{\sum_l t_{il, n}v_{lj, n}}
  +  \log \sum_l t_{il, n}v_{lj, n} \right]   \nonumber \\
    & \quad - 2J \sum_{i} \log |\det \bm{W}_i|, \label{eq:ilrma1}
\end{align}
where $\mathsf{W}$, $\mathsf{T}$, and $\mathsf{V}$ represent the sets of
$\bm{W}_i$, $\bm{T}_n$, and $\bm{V}_n$ for all $i$ and $n$,
and $t_{il,n}$ and $v_{lj,n}$ are the nonnegative elements of $\bm{T}_{n}$ and $\bm{V}_{n}$, respectively.
The rank-$L$ matrix $\bm{T}_{n}\bm{V}_{n}$ corresponds to the NMF decomposition
and represents an estimated power spectrogram of the $n$th source.

The first and third terms in \eqref{eq:ilrma1} are
equivalent to the cost function in IVA, which evaluates the independence between sources,
and the first and second terms in \eqref{eq:ilrma1}  are equivalent to the cost function in simple ISNMF.
Efficient update rules for optimizing $\mathcal{J}(\mathsf{W},\mathsf{T},\mathsf{V})$
based on the MM algorithm were proposed in \cite{Kitamura2016}.

\begin{figure}[!t]
\centering
\includegraphics[width=0.95\hsize]{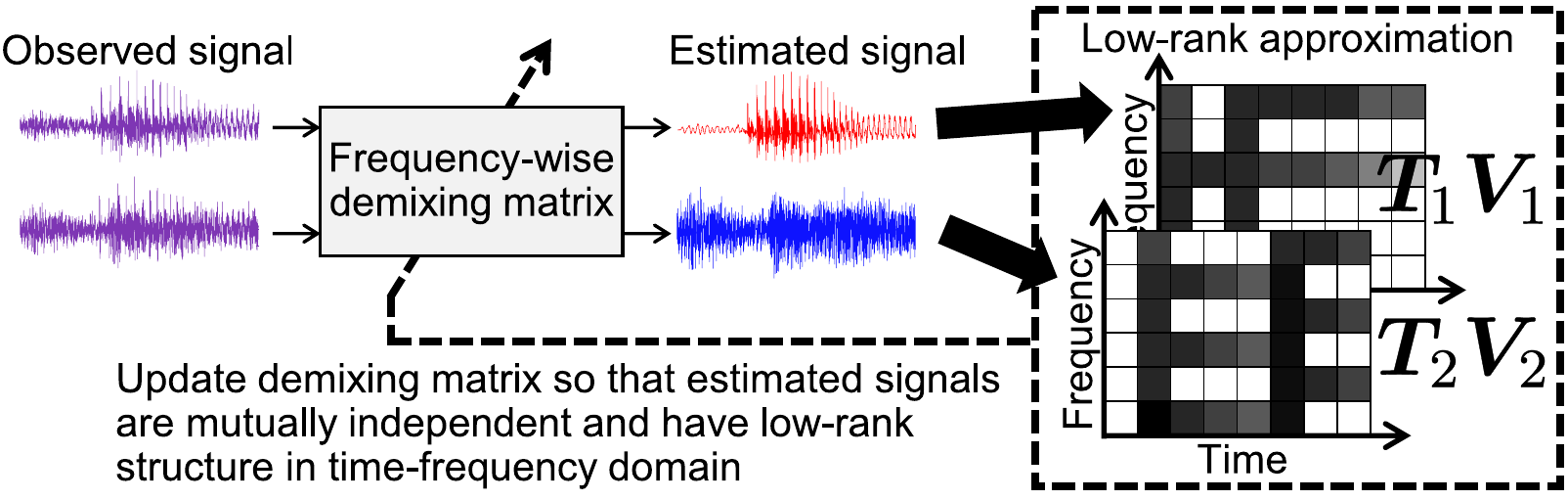}
\vspace{-1mm}
\caption{Conceptual model of ILRMA.}
\label{fig:ILRMA}
\vspace{-2mm}
\end{figure}

\section{Proposed Method}

\subsection{Motivation}
In ILRMA, we alternately update the demixing matrix $\bm{W}_i$ and the NMF variables $\bm{T}_n$ and $\bm{V}_n$.
It is empirically known that the convergence speed of the demixing matrix is much higher than that of the
NMF variables, and this difference may result in unstable separation owing to the existence of local minima.
To avoid convergence to poor solutions, we introduce a new parameter $p \in (0,1]$
that controls the convergence speed of the NMF variables.
Also, we experimentally investigate the best balance between the convergence speeds
of the spatial and source models in ILRMA to achieve more accurate separation.

\subsection{Auxiliary-Function-Based Iterative Optimization}
In this subsection, we briefly introduce the mechanism in optimization based on the auxiliary function technique,
which includes the MM and ME algorithms.
We consider an optimization problem to minimize a cost function $Q(\bm{\theta})$
that is difficult to directly minimize, where $\bm{\theta}$ belongs to an arbitrary parameter space $S$.
In such a case, we can use an auxiliary-function-based iterative optimization approach.
In this method, we first define an auxiliary function of $Q(\bm{\theta})$ as $Q^{+}(\bm{\theta}|\tbm{\theta})$ that satisfies
\begin{itemize}
    \item $Q(\bm{\theta}) = Q^{+}(\bm{\theta}|\bm{\theta})$ for all $\bm{\theta} \in S$,
    \item $Q(\bm{\theta}) \leq Q^{+}(\bm{\theta}|\tbm{\theta})$ for all $\bm{\theta}, \tbm{\theta} \in S$,
\end{itemize}
where $\tbm{\theta}$ is called the auxiliary variable for $\bm{\theta}$.
In the MM algorithm, we replace the minimization of $Q(\bm{\theta})$
with the minimization of the auxiliary function $Q^{+}(\bm{\theta}|\tbm{\theta})$.
The parameter $\bm{\theta}$ is updated as
\begin{align}
    \bm{\theta}^{(c+1)} = \argmin_{\bm{\theta}' \in S} Q^{+}(\bm{\theta}'|\bm{\theta}^{(c)}),
\end{align}
where $\bm{\theta}^{(c)}$ is the parameter at the $c$th iteration.
On the other hand, in the ME algorithm,
we find a parameter $\bm{\theta}^{(c+1)}$ that satisfies
\begin{align}
    Q^{+}(\bm{\theta}^{(c+1)}|\bm{\theta}^{(c)}) =  Q^{+}(\bm{\theta}^{(c)}|\bm{\theta}^{(c)}). \label{eq:MEkimo}
\end{align}
Fig.~\ref{fig:auxImage} shows geometric interpretations of the MM and ME algorithms.
Both the MM and ME algorithms guarantee the monotonic nonincrease in $Q(\bm{\theta})$.

\begin{figure}[!t]
\centering
\includegraphics[width=0.80\columnwidth]{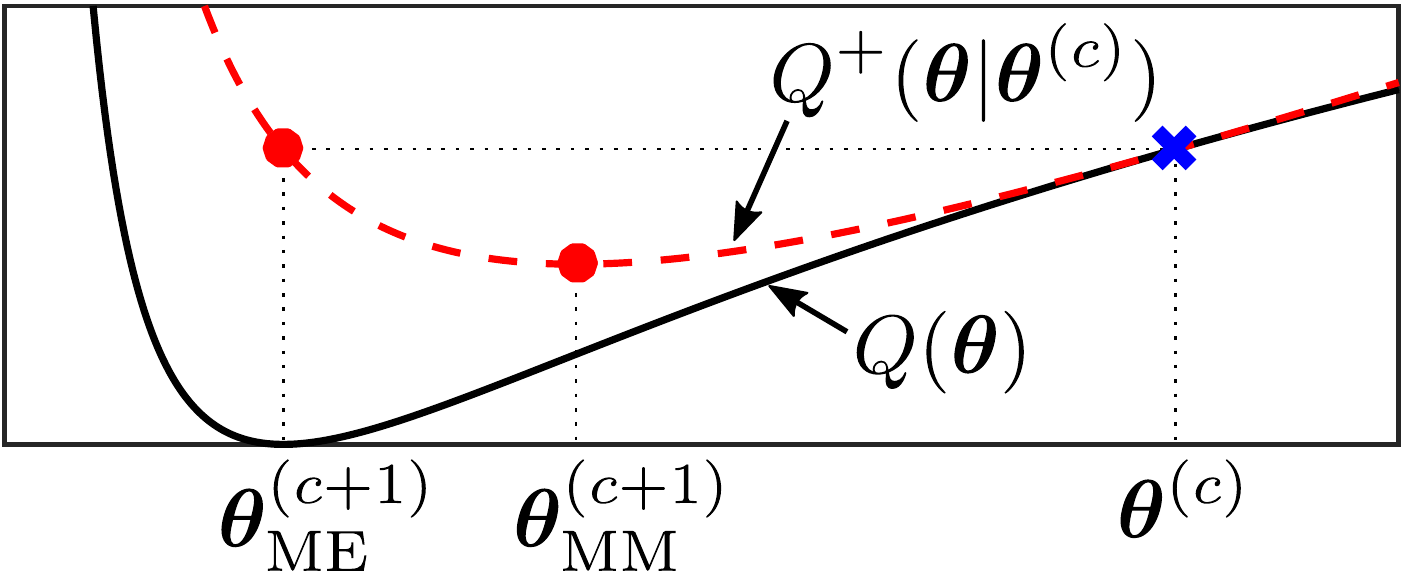}
\vspace{-1mm}
\caption{Geometric interpretations of MM and ME algorithms, where $\bm{\theta}_{\mathrm{MM}}^{(c+1)}$ and $\bm{\theta}_{\mathrm{ME}}^{(c+1)}$ are parameters in $(c+1)$th iteration with MM and ME algorithms,
respectively.}
\label{fig:auxImage}
\vspace{-2mm}
\end{figure}

\subsection{Derivation of New Parametric Auxiliary Function}
To control the convergence speed of the NMF variables in ILRMA, we design a new parametric auxiliary function
with the parameter $p$.

First, we give a brief proof of the following lemma.
\begin{Lem} 
For all $0 < q_1 < q_2$ or $q_2 < q_1 < 0$,
\begin{align}
    z^{q_1} \leq \frac{q_1}{q_2} z^{q_2} \xi^{{q_1}-{q_2}} + \frac{{q_2}-{q_1}}{{q_2}} \xi^{{q_1}} \label{eq:shisu}
\end{align}
holds, where $z>0$ and $\xi>0$.
The equality in \eqref{eq:shisu} holds if and only if $z =\xi$.
\end{Lem}
\begin{proof}
$u^{q_1/q_2}$ is a concave function for $u>0$ because $0 < q_1/q_2 < 1$. From the tangent-line inequality, we obtain
\begin{align}
    u^{\frac{q_1}{q_2}} &\leq \frac{q_1}{q_2} \eta^{\frac{q_1}{q_2}-1} (u-\eta) + \eta^{\frac{q_1}{q_2}}, \label{eq:tan}
\end{align}
where $\eta > 0$, and the equality holds if and only if $u = \eta$.
By substituting $u = z^{q_2}$ and $\eta = \xi^{q_2}$ and simplifying \eqref{eq:tan}, we have \eqref{eq:shisu}.
The equality in \eqref{eq:shisu} holds if and only if  $z = \xi$.
\end{proof}

Fig.~\ref{fig:aux} shows a geometric interpretation of \eqref{eq:shisu}.
When $q_2 = 1$, \eqref{eq:shisu} is equivalent to the tangent-line inequality, which is commomly used for designing auxiliary functions of concave terms.
In contrast, owing to Lemma 1, we can find various types of upper-bound function
except for the case of $q_2 = 1$, such as for $q_2 = 8$, $4$, $2$, and $0.7$.
This property implies the possibility of constructing different types of auxiliary function,
as described below.

\begin{figure}[!t]
\centering
\includegraphics[width=0.80\columnwidth]{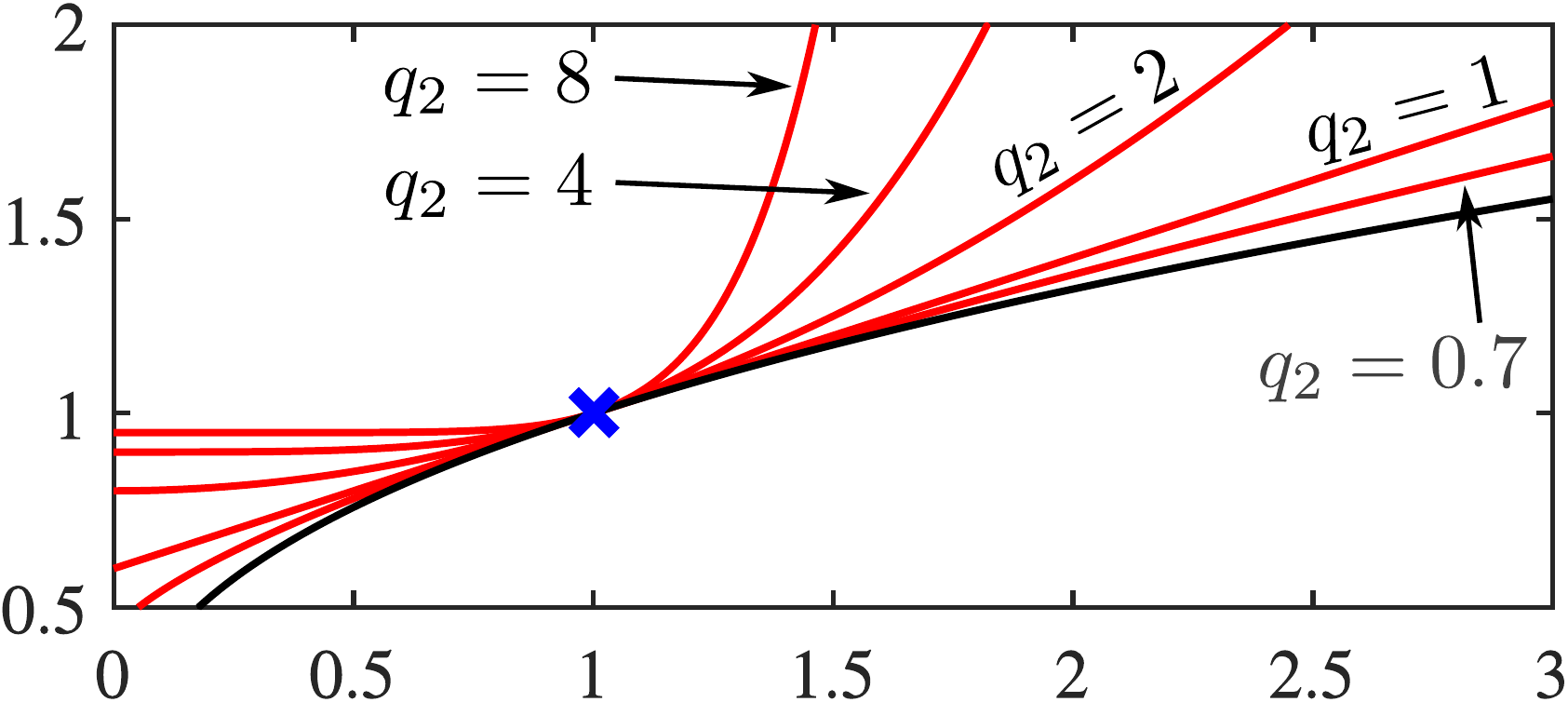}
\caption{Geometric interpretation of \eqref{eq:shisu}, where $q_1 = 0.4$ and $\xi = 1$.
Left- and right-hand sides of \eqref{eq:shisu} are depicked by black and red lines, respectively.}
\vspace{-1mm}
\label{fig:aux}
\vspace{-2mm}
\end{figure}

Next, we design the auxiliary function of \eqref{eq:ilrma1}.
The first term in \eqref{eq:ilrma1} is a convex function and the second term is a concave function for
$t_{il,n}v_{lj,n}$.
Applying Jensen's inequality to the first term using an auxiliary variable $\alpha_{ijl,n} \geq 0$ that
satisfies $\sum_l \alpha_{ijl,n} = 1$, we have
\begin{align}
    \frac{1}{\sum_{l} t_{il,n}v_{lj,n}} \leq \sum_{l} \frac{\alpha_{ijl,n}^{2}}{t_{il,n}v_{lj,n}}. \label{eq:jensen}
\end{align}
Also, by applying the tangent-line inequality to the second term using an auxiliary variable $\beta_{ij,n} > 0$, we have
\begin{align}
    \log \sum_l t_{il,n} v_{lj,n} \leq \frac{1}{\beta_{ij,n}} \left( \sum_l t_{il,n} v_{lj,n} - \beta_{ij,n} \right)
    + \log \beta_{ij,n}. \label{eq:tangent}
\end{align}
The equalities in \eqref{eq:jensen} and \eqref{eq:tangent} hold if and only if
\begin{align}
    \alpha_{ijl,n} &= \frac{t_{il,n}v_{lj,n}}{\sum_{l'} t_{il',n} v_{l'j,n}}, \label{eq:aux1} \\
    \beta_{ij,n} &= \sum_{l} t_{il,n} v_{lj,n}.  \label{eq:aux2}
\end{align}
From \eqref{eq:ilrma1} and \eqref{eq:jensen}--\eqref{eq:aux2}, we have the following auxiliary function:
\begin{align}
     \mathcal{J}^{+}(\mathsf{W},\mathsf{T},\mathsf{V} | \tsf{T}, \tsf{V}) &= \sum_{i, j, l, n}
    \left[
      \frac{ |\ctranspose{\bm{w}_{i,n}}\bm{x}_{ij} |^2 \tilde{t}_{il,n} \tilde{v}_{lj,n} }{ (\sum_{l'} \tilde{t}_{il',n}\tilde{v}_{l'j,n})^2 } \cdot
       \frac{\tilde{t}_{il,n} \tilde{v}_{lj,n}}{t_{il,n}v_{lj,n}}
    \right. \nonumber \\
    & \quad \left. {} +  \frac{\tilde{t}_{il,n}\tilde{v}_{lj,n}}{\sum_{l'} \tilde{t}_{il',n}\tilde{v}_{l'j,n}}
    \cdot  \frac{t_{il,n}v_{lj,n}}{\tilde{t}_{il,n} \tilde{v}_{lj,n}}
    \right] + \mathcal{C}, \label{eq:ilrma1-aux}
\end{align}
where we denote the auxiliary variables using $\tilde{\cdot}$ to distinguish them from the original variables,
and $\mathcal{C}$ is a term independent of $\mathsf{T}$ and $\mathsf{V}$.
This auxiliary function was used in~\cite{Kitamura2016} to derive the update rules for $\mathsf{T}$ and $\mathsf{V}$
based on the MM algirithm.
In this paper, we design a further auxiliary function parameterized by $p$
to control the convergence speed of the update rules.
By applying \eqref{eq:shisu} to $(t_{il,n}v_{lj,n})^{-1}$ and $t_{il,n}v_{lj,n}$ in the right-hand side of
\eqref{eq:ilrma1-aux} using auxiliary variables $\gamma_{ijl,n} > 0$ and $\delta_{ijl,n} > 0$, we have
\begin{align}
    (t_{il,n}v_{lj,n})^{-1} &\leq p  \gamma_{ijl,n}^{-1+\frac{1}{p}} \left( t_{il,n}v_{lj,n} \right)^{-\frac{1}{p}} + (1 - p) \gamma_{ijl,n}^{-1}, \label{eq:shisu2} \\
    t_{il,n}v_{lj,n} &\leq p  \delta_{ijl,n}^{1-\frac{1}{p}} \left( t_{il,n}v_{lj,n} \right)^{\frac{1}{p}} + (1 - p) \delta_{ijl,n}. \label{eq:shisu3}
\end{align}
The equalities in \eqref{eq:shisu2} and \eqref{eq:shisu3} hold if and only if
\begin{align}
    \gamma_{ijl,n} &= \delta_{ijl,n} = t_{il,n} v_{lj,n}. \label{eq:aux3}
\end{align}
From \eqref{eq:ilrma1-aux}--\eqref{eq:aux3}, we have
\begin{align}
    & \mathcal{J}_p^{++}(\mathsf{W},\mathsf{T},\mathsf{V} | \tsf{T}, \tsf{V}) \nonumber \\
    &= p  \sum_{i, j, l, n}  \left[  \frac{ |\ctranspose{\bm{w}_{i,n}}\bm{x}_{ij} |^2 \tilde{t}_{il,n} \tilde{v}_{lj,n}}{(\sum_{l'} \tilde{t}_{il',n}\tilde{v}_{l'j,n})^2 } \cdot \left( \frac{\tilde{t}_{il,n} \tilde{v}_{lj,n}}{t_{il,n}v_{lj,n}}   \right)^{\frac{1}{p}} \right. \nonumber \\
    & \quad \left. {} +  \frac{\tilde{t}_{il,n}\tilde{v}_{lj,n}}{\sum_{l'} \tilde{t}_{il',n}\tilde{v}_{l'j,n}}
    \left( \frac{t_{il,n}v_{lj,n}}{\tilde{t}_{il,n} \tilde{v}_{lj,n}}   \right)^{\frac{1}{p}}
    \right] + \mathcal{C}', \label{eq:ilrma1-aux2}
\end{align}
where $\mathcal{C}'$ is a term independent of $\mathsf{T}$ and $\mathsf{V}$.
The inequality among \eqref{eq:ilrma1}, \eqref{eq:ilrma1-aux}, and \eqref{eq:ilrma1-aux2} can be represented as
\begin{align}
  &\mathcal{J}(\mathsf{W},\mathsf{T},\mathsf{V}) \leq
  \mathcal{J}^{+} (\mathsf{W},\mathsf{T},\mathsf{V}|\tsf{T},\tsf{V})
   \leq \mathcal{J}_p^{++} (\mathsf{W},\mathsf{T},\mathsf{V}|\tsf{T},\tsf{V}), \label{eq:kekka}
\end{align}
where the equality in \eqref{eq:kekka} holds if and only if $\tsf{T} = \mathsf{T}$ and $\tsf{V} = \mathsf{V}$.
Therefore, when $p$ approaches 1, $\mathcal{J}_p^{++}$ equals $\mathcal{J}^{+}$.
On the other hand, as $p$ approaches 0, $\mathcal{J}_p^{++}$ diverges from
the original cost function $\mathcal{J}$, which leads to slower update rules for the NMF variables
as described in the next section.

\subsection{Derivation of New Update Rules for ILRMA}
The original cost function \eqref{eq:ilrma1} can be minimized by employing the parametric auxiliary function
\eqref{eq:ilrma1-aux2} based on the ME algorithm to derive the update rules for $\mathsf{T}$ and $\mathsf{V}$.
For the variables in $\mathsf{T}$, we have to solve the following equation:
\begin{align}
    \mathcal{J}_p^{++} (\mathsf{W},\mathsf{T},\tsf{V}|\tsf{T},\tsf{V}) =
    \mathcal{J}_p^{++} (\mathsf{W},\tsf{T},\tsf{V}|\tsf{T},\tsf{V}). \label{eq:me}
\end{align}
By taking \eqref{eq:me} apart and organizing the terms of $t_{il,n}^{-\frac{1}{p}}$ and $t_{il,n}^{\frac{1}{p}}$, we obtain
\begin{align}
    A \left( \frac{\tilde{t}_{il,n}}{t_{il,n}}  \right)^{\frac{1}{p}} +
    B \left( \frac{t_{il,n}}{\tilde{t}_{il,n} }  \right)^{\frac{1}{p}}
    + \mathcal{C}''
    = A + B + \mathcal{C}'', \label{eq:niji}
\end{align}
where $A = \sum_{j} |\ctranspose{\bm{w}_{i,n}}\bm{x}_{ij} |^2  \tilde{v}_{lj,n}
(\sum_{l'} \tilde{t}_{il',n}\tilde{v}_{l'j,n})^{-2}$,
$B = \sum_{j} \tilde{v}_{lj,n} (\sum_{l'} \tilde{t}_{il',n}\tilde{v}_{l'j,n})^{-1}$,
and $\mathcal{C}''$ is a term independent of $t_{il,n}$.
Equation \eqref{eq:niji} is a quadratic equation in the variable $\left( t_{il,n} / \tilde{t}_{il,n}  \right)^{\frac{1}{p}}$,
and we can factorize \eqref{eq:niji} as
\begin{align}
    \left[  \left( \frac{t_{il,n}}{\tilde{t}_{il,n} }  \right)^{\frac{1}{p}} - \frac{A}{B} \right]
    \left[  \left( \frac{t_{il,n}}{\tilde{t}_{il,n} }  \right)^{\frac{1}{p}} - 1 \right] \label{eq:niji-fac}
    = 0.
\end{align}
Finally, we can obtain the following update rule for $t_{il,n}$ as the nontrivial solution of \eqref{eq:niji-fac},
i.e., $t_{ij,n} = \tilde{t}_{ij,n} (A/B)^p$:
\begin{align}
    t_{il,n} = \tilde{t}_{il,n}  \left[ \frac{\sum_j |\ctranspose{\bm{w}_{i,n}}\bm{x}_{ij} |^2  \tilde{v}_{lj,n} (\sum_{l'} \tilde{t}_{il',n}\tilde{v}_{l'j,n})^{-2} }{ \sum_j \tilde{v}_{lj,n} (\sum_{l'} \tilde{t}_{il',n}\tilde{v}_{l'j,n})^{-1}} \right]^{p}. \label{eq:up_t}
\end{align}
Similar to \eqref{eq:up_t}, we can derive the update rule for $v_{lj,n}$ as
\begin{align}
    v_{lj,n} &= \tilde{v}_{lj,n}  \left[ \frac{\sum_i |\ctranspose{\bm{w}_{i,n}}\bm{x}_{ij} |^2  \tilde{t}_{il,n} (\sum_{l'} \tilde{t}_{il',n}\tilde{v}_{l'j,n})^{-2} }{ \sum_i \tilde{t}_{il,n} (\sum_{l'} \tilde{t}_{il',n}\tilde{v}_{l'j,n})^{-1}} \right]^{p} . \label{eq:up_v}
\end{align}
These update rules coincide with those in conventional ILRMA when $p = 0.5$.
By controlling this parameter in the range of $0 < p \leq 1$, we can adjust the convergence speed
of the NMF variables $\bm{T}_n$ and $\bm{V}_n$.
Note that these update rules guarantee the monotonic nonincrease in $\mathcal{J}$
for any value of $p$ in the range $0 < p \leq 1$.
The similar parameterization of the MM algorithm can be derived.
However, it just provides the same generalization in the range of $0<p\leq0.5$.
Regarding the demixing matrix $\bm{W}_i$, we can use the same update rule as that in conventional ILRMA
called iterative projection, which was originally proposed in AuxIVA~\cite{Ono2011}.

\section{Experiment}

\subsection{Experimental Setup}

\begin{table}[!t]
\centering
\caption{Music sources}
\vspace{-2mm}
\label{tbl:sources}
\begin{tabular}{@{}ccc@{}}
\toprule
ID & Song name & Instruments \\ \midrule
1 & dev2\_\_ultimate\_nz\_tour\_\_snip\_43\_61           & guitar/synth                  \\
2 & dev2\_\_another\_dreamer-the\_ones\_we\_love\_\_snip\_69\_94 & guitar/vocals \\
 \bottomrule
\end{tabular}
\end{table}

\begin{table}[!t]
\centering
\caption{Experimental conditions}
\vspace{-2mm}
\label{tbl:condition}
\begin{tabular}{@{}cc@{}}
\toprule
Sampling frequency             & 16000~Hz                     \\
FFT length                     & 256~ms (4096~samples) \\
Window shift                   & 128~ms (2048~samples)  \\
Number of bases $L$        & \{10, 20, 40, 60\}         \\
Number of ILRMA iterations & 200                              \\
Initial values of $\bm{T}_n, \bm{V}_n$ &  Uniform random values in $(0,1)$ \\
Initial value of $\bm{W}_i$ &  $N \times N$ identity matrix    \\
 \bottomrule
\end{tabular}
\vspace{-4mm}
\end{table}

To find the best balance between the convergence speeds,
we conducted BSS experiments with two sources and two microphones.
In this experiment, the observed signals were convoluted with the music sources shown in Table~\ref{tbl:sources}
obtained from the SiSEC2010 professionally produced music recordings dataset~\cite{Araki2010}
and the impulse response (E2A)~\cite{Nakamura2000} used in~\cite{Kitamura2016}.
Also, we used uniformly produced random values of $p \in (0,1]$ in the proposed parametric update rules
\eqref{eq:up_t} and \eqref{eq:up_v}.
As the evaluation criterion, we used the signal-to-distortion ratio (SDR)~\cite{Vincent2006},
which indicates the overall separation quality.
We performed ILRMA with the proposed parametric update rules 2000 times using different random seeds
in each experiment.
The other conditions are shown in Table~\ref{tbl:condition}.

\subsection{Results}

\begin{figure}[!t]
\centering
\subfloat[\label{fig:rand1}]{\includegraphics[clip, width=0.83\columnwidth]{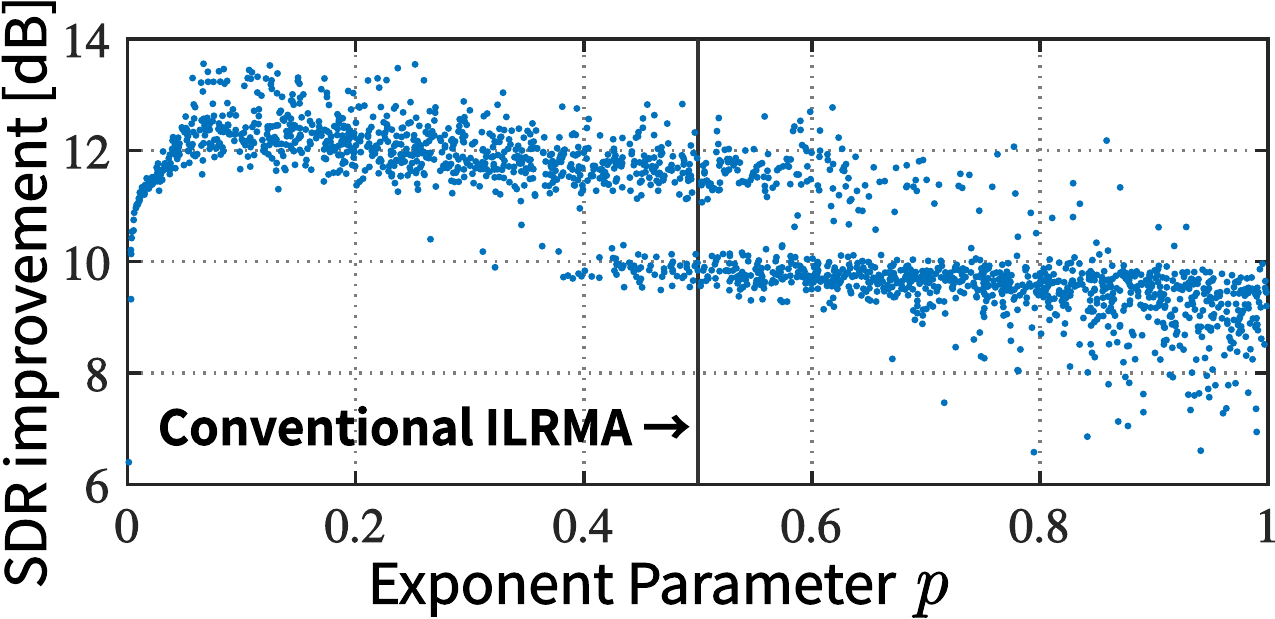}}
\quad
\subfloat[\label{fig:rand2}]{\includegraphics[clip, width=0.83\columnwidth]{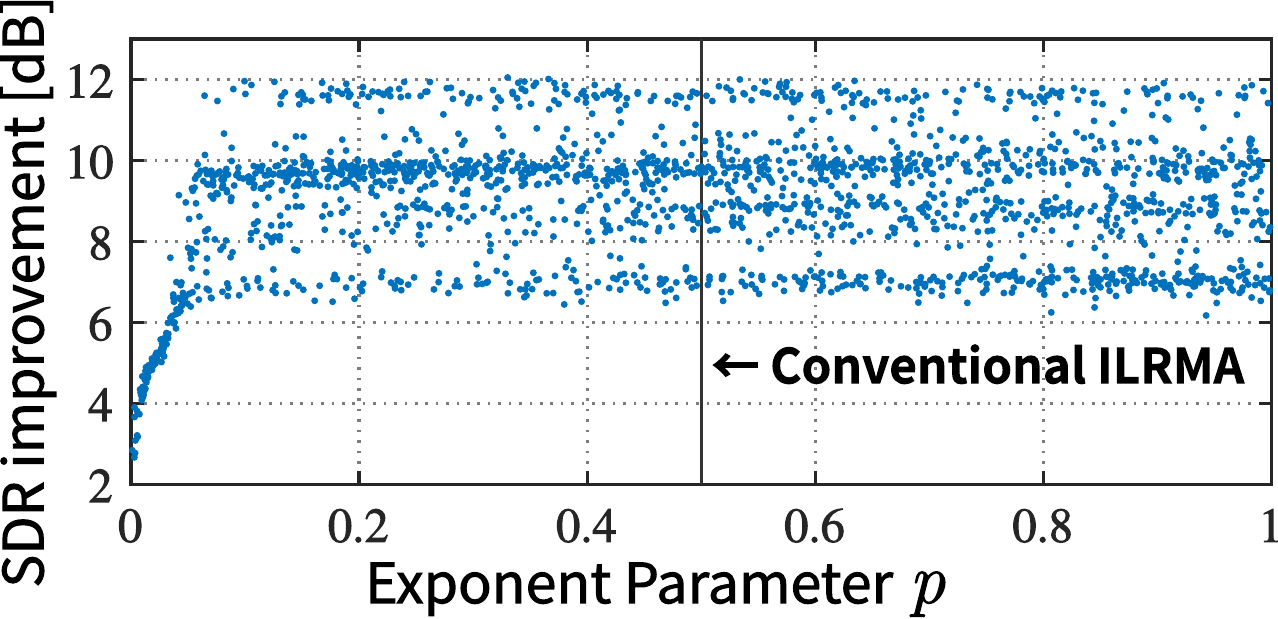}}
\vspace{-1.5mm}
\caption{Scatter plots of SDR improvements of the guitar source in (a) song ID1 and (b) song ID2, where $L = 40$.}
\label{fig:rand}
\vspace{-2mm}
\end{figure}

\begin{figure}[!t]
\centering
\subfloat[ID1]{\includegraphics[clip, width=0.48\columnwidth]{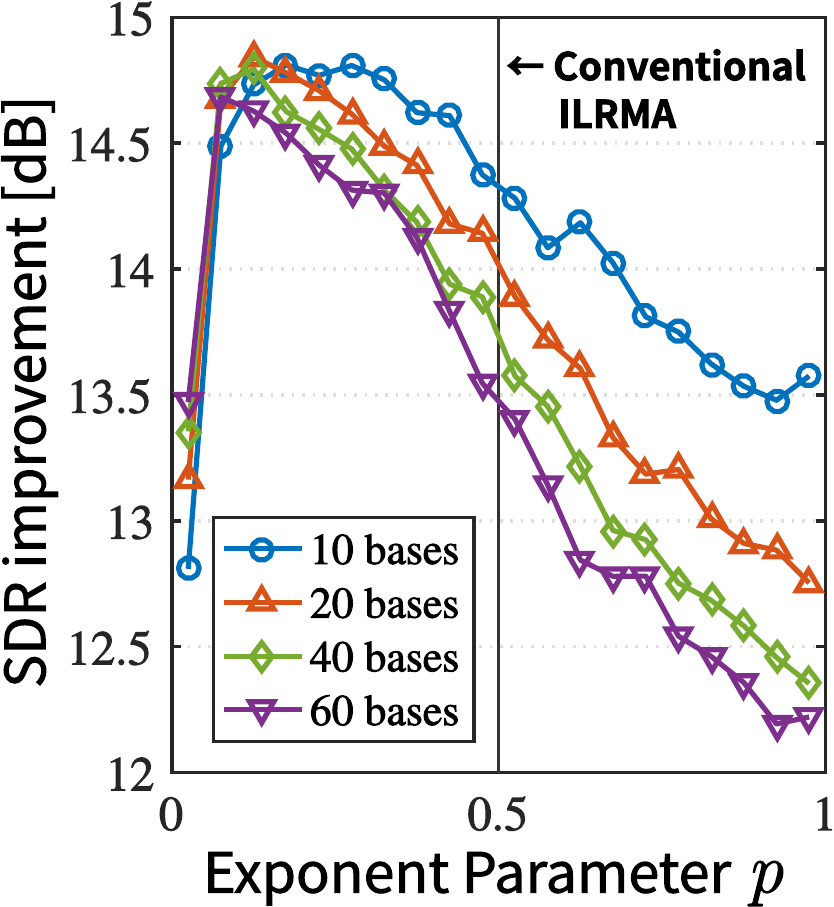}}
\quad
\subfloat[ID2]{\includegraphics[clip, width=0.48\columnwidth]{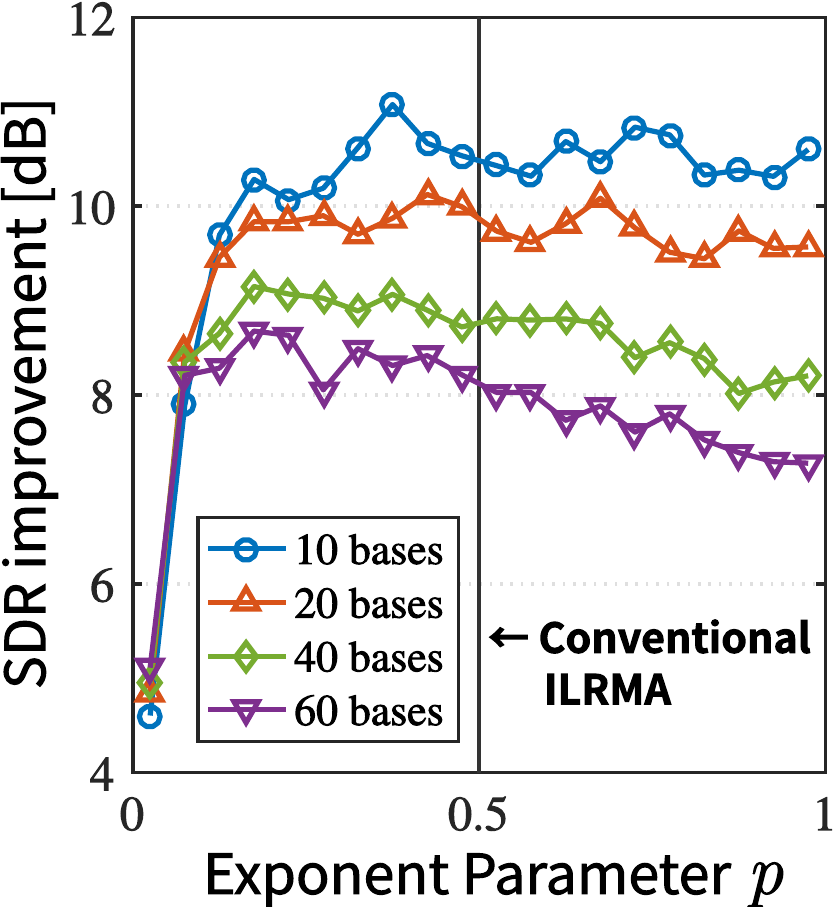}}
\vspace{-1.5mm}
\caption{Average SDR improvements of (a) song ID1 and (b) song ID2.}
\label{fig:sound}
\vspace{-3mm}
\end{figure}

Fig.~\ref{fig:rand} shows the scatter plot of SDR improvements for various values of  $p$.
From Fig.~\ref{fig:rand}~\subref{fig:rand1}, we can confirm that slower update rules for the NMF variables
(smaller values of $p$) provide better separation performance than faster ones (larger values of $p$).
This tendency did not clearly appear in Fig.~\ref{fig:rand}~\subref{fig:rand2},
but we can still recognize that the cluster in $p \in (0, 0.5)$ is concentrated around higher SDR values.
Fig.~\ref{fig:sound} shows the average SDR improvements calculated from the 2000 separation results.
When we set $L$ to a larger value, the performance of ILRMA with $p \in (0.5, 1)$ is degraded.
However, by using a smaller vallue of $p$,
ILRMA can achieve better separation even if too many bases are used in the source model,
thus avoiding a poor solution.

\section{Conclusion}
In this paper, we proposed a new optimization method of ILRMA based on the parametric ME algorithm
to control the optimization balance between spatial and source models, where a monotonic nonincrease in
the cost function is always ensured.
The experimental results showed that ILRMA with a smaller value of $p$ tends to provide better separation
than conventional ILRMA, especially in the case of a large number of bases.

\section*{Acknowledgments}
This work was partly supported by the SECOM Science and Technology Foundation and JSPS KAKENHI Grant Numbers JP17H06101 and JP17H06572.

\IEEEtriggeratref{13}


\bibliographystyle{IEEEtran}
\bibliography{bstctl,abbrev_middle,paper}
%

%
%

\end{document}